\documentclass[review,12pt]{elsarticle}
\usepackage{amsmath}
\usepackage{amsfonts}
\usepackage{amssymb}
\usepackage{bbm}
\usepackage{bbold}
\usepackage{amsthm}
\theoremstyle{plain}
\newtheorem{theorem}{Theorem} 

\newtheorem{proposition}[theorem]{Proposition}

\newtheorem{definition}[theorem]{Definition}

\newtheorem{example}[theorem]{Example}

\journal{Journal of Mathematical Psychology Templates}

\biboptions{authoryear}

\begin{document}

\begin{frontmatter}

\title{An invitation to coupling and copulas:\\\large with applications to multisensory modeling} %

\author{Hans Colonius\fnref{myfootnote}}
\address{University of Oldenburg, Oldenburg}
\fntext[myfootnote]{E-mail address: hans.colonius@uni-oldenburg.de\\\ead[url]{www.uni-oldenburg.de/en/hans-colonius/}}




\begin{abstract}
This paper presents an introduction to the stochastic concepts of \emph{coupling} and \emph{copula}. Coupling means the construction of a joint distribution of two or more random variables that need not be defined on one and the same probability space, whereas a copula is a function that joins a multivariate distribution to its one-dimensional margins. Their role in stochastic modeling is illustrated by examples from multisensory perception. Pointers to more advanced and recent treatments are provided.
\end{abstract}


\end{frontmatter}


\section{Introduction}
The concepts of \emph{coupling} and \emph{copula} refer to two related areas of probability and statistics whose importance for mathematical psychology has  arguably been ignored so far.  This paper gives an elementary, non-rigorous introduction to both concepts.  Moreover, applications of both concepts to modeling in a multisensory context are presented. Briefly, \textit{coupling} means the construction of a joint distribution of two or more random variables that need not be defined on one and the same probability space, whereas a \textit{copula} is a function that joins a multivariate distribution to its one-dimensional margins. 

Empirical psychological data are typically collected under various experimental conditions, e.g. speed vs. accuracy instructions in a reaction time (RT) experiment, different numbers of targets and nontargets in a visual search paradigm, or hits and false alarms in a detection task. Data collected in any particular condition are considered as realizations of some random variable with respect to an underlying probability (sample) space. Note that there is no principled way of stochastically relating random variables observed under different conditions: for example, an observed time to find a target in a condition with $n$ targets, $T_n$, to that for a condition with $n+1$ targets, $T_{n+1}$. The reason simply is that $T_n$ and $T_{n+1}$ cannot be observed within the same trial, that is, they do not refer to (elementary) outcomes defined in the same probability space. This does not rule out numerically comparing average data under both conditions, or even the entire distributions functions; however, any statistical hypothesis, e.g. about the correlation between $T_n$ and $T_{n+1}$, would be void. A coupling construction, on the other hand, could allow for just that. For example, as shown below, one can turn a statement about an ordering of two RT distributions, which are a-priori not stochastically related at all, into a statement about point-wise ordering of the corresponding random variables on a common probability space. The choice of a particular coupling construction, however, is somewhat arbitrary and may be more or less useful, depending on the specific goals of the modeler. It should be noted that coupling turns out to be a key concept in a general framework for ``contextuality'' being developed by E. N. Dzhafarov and colleagues \citep{Dzhafarovinpress}.

The concept of copula has stirred a lot of interest in recent years in several areas of statistics, including finance, mainly for the following reasons \citep[see e.g.,][]{joe2015}: it allows one (i)~to study the structure of stochastic dependency in a ``scale-free'' manner, i.e., independent of the specific marginal distributions, and (ii)~to construct families of multivariate distributions with specified properties. We will demonstrate in the final section how copulas can be used to test models of multisensory integration and to derive measures of the amount of multisensory integration occurring in a given context.

\section{Coupling}
We begin with a few common definitions\footnote{If not indicated otherwise, most of the material on coupling in this section is taken from the monograph by \cite{thorisson2000}.}. Let $X$ and $Y$ be random variables defined on a probability space $(\Omega, \mathcal{F}, \mathbb{P})$, with values in $(\mathbb{R},\mathcal{B},P)$. ``Equality'' of random variables can be interpreted in different ways. Random variables $X$ and $Y$ are \textit{equal in distribution} iff they are governed by the same probability measure:
\[X =_d Y \; \Longleftrightarrow \; P(X\in A) =P(Y\in A)  \; \text{, for all $A \in \mathcal{B}$}.\] 
$X$ and $Y$ are \textit{point-wise equal} iff they agree for almost all elementary events\footnote{Here, a.s. is for ``almost surely''.}:
\[X \overset{a.s.}{=} Y \; \Longleftrightarrow \;P(\{\, \omega \mid X(\omega)=Y(\omega)\})=1.\]
Let $X$ be a real-valued random variable with distribution function $F(x)$. Then the \textit{quantile function} of $X$ is defined as 
\begin{equation}
Q(u)=F^{-1}(u)=\inf\{x\,|\,F(x)\ge u\}, \;\;\;0\le u \le 1.
\end{equation}
For every $-\infty <x<+\infty$ and $0<u<1$, we have
\[ F(x) \ge u \;\; \mbox{ if, and only if, } Q(u) \le x. \]
Thus, if there exists $x$ with  $F(x)=u$, then $F(Q(u))=u$ and $Q(u)$ is the smallest value of $x$ satisfying $F(x)=u$. If $F(x)$ is continuous and strictly increasing, $Q(u)$ is the unique value $x$ such that $F(x)=u$.

Moreover, if $U$ is a standard uniform random variable (i.e., defined on $[0,1]$), then $X=Q(U)$ has its distribution function as $F(x)$; thus, any distribution function can be conceived as arising from the uniform distribution transformed by $Q(u)$.
\subsection{Definition and examples}
\begin{quote}
\textbf{Note to the reader: We enumerate Definitions, Theorems, Examples  sequentially, so Definition 1 is followed by Example 2, and so on.
}\end{quote}
\begin{definition}
A \textit{coupling} of a collection of random variables $\{X_i, i \in \mathrm{I}\}$, with $\mathrm{I}$ denoting some index set, is a family of jointly distributed random variables 
\[(\hat{X}_i: \, i \in \mathrm{I}) \;\; \mbox{ such that } \; \;  \hat{X}_i =_dX_i,\;\; i \in \mathrm{I}. \]
\end{definition}
Note that the joint distribution of the $\hat{X}_i$ need not be the same as that of  $X_i$; in fact, the $X_i$ may not  even have a joint distribution because they need not be defined on a common probability space. However, the family $(\hat{X}_i: \, i \in \mathrm{I})$ has a joint distribution with the property that its marginals are equal to the distributions of the individual $X_i$ variables. The individual $\hat{X}_i$ is also called a \textit{copy of} $X_i$.
\begin{example}[Coupling two Bernoulli random variables]
Let $X_p, X_q$ be  Bernoulli random variables, i.e.,  \[P(X_p=1)=p \;\mbox{ and }\; P(X_p=0)=1-p,\] and $X_q$ defined analogously. Assume $p<q$; we can couple $X_p$ and $X_q$ as follows: 
\newline Let $U$ be a uniform random variable on $[0, 1]$, i.e., for $0\le a < b \le 1$,
\[  P(a < U \le b)= b-a.\]  Define 
\begin{align*}
 \hat{X}_p =
\begin{cases}
1, &\text{if $0 <U\le p$;} \\
0, & \text{if $p < U \le 1$}
\end{cases} ; \;\; &\;\; 
 \hat{X}_q =
\begin{cases}
1, &\text{if $0 <U\le q$;} \\ 
0, & \text{if $q < U \le 1$}.
\end{cases}
\end{align*}
Then $U$ serves as a common source of randomness for both $\hat{X}_p$ and $\hat{X}_q$. Moreover, $\hat{X}_p =_dX_p$ and $\hat{X}_q =_dX_q$, and $\mathrm{Cov}(\hat{X}_p,\hat{X}_q)=p (1-q)$. The joint distribution of $(\hat{X}_p,\hat{X}_q)$ is presented  in Table~\ref{bernoulli} and illustrated in Figure~\ref{bernoullidist}.
\begin{table}[htbp]
\caption{Joint distribution of two Bernoulli random variables.}
\label{bernoulli}
\begin{center}
\begin{tabular}{ccccc}
\multicolumn{5}{c}{$\hat{X}_q$} \\ 
& & 0 &  1 & \\ \hline
& 0 & $1-q$ & $q-p$ & $1-p$\\ 
\raisebox{1.5 ex} [-1.5ex]{$\hat{X}_p$} & 1 & \hfill 0 \hfill  & p & p \\ \hline
& & $1-q$ & $q$ & 1

\end{tabular}
\end{center}
\end{table}
\begin{figure}[htbp]
\begin{center}
\includegraphics[scale=0.28]{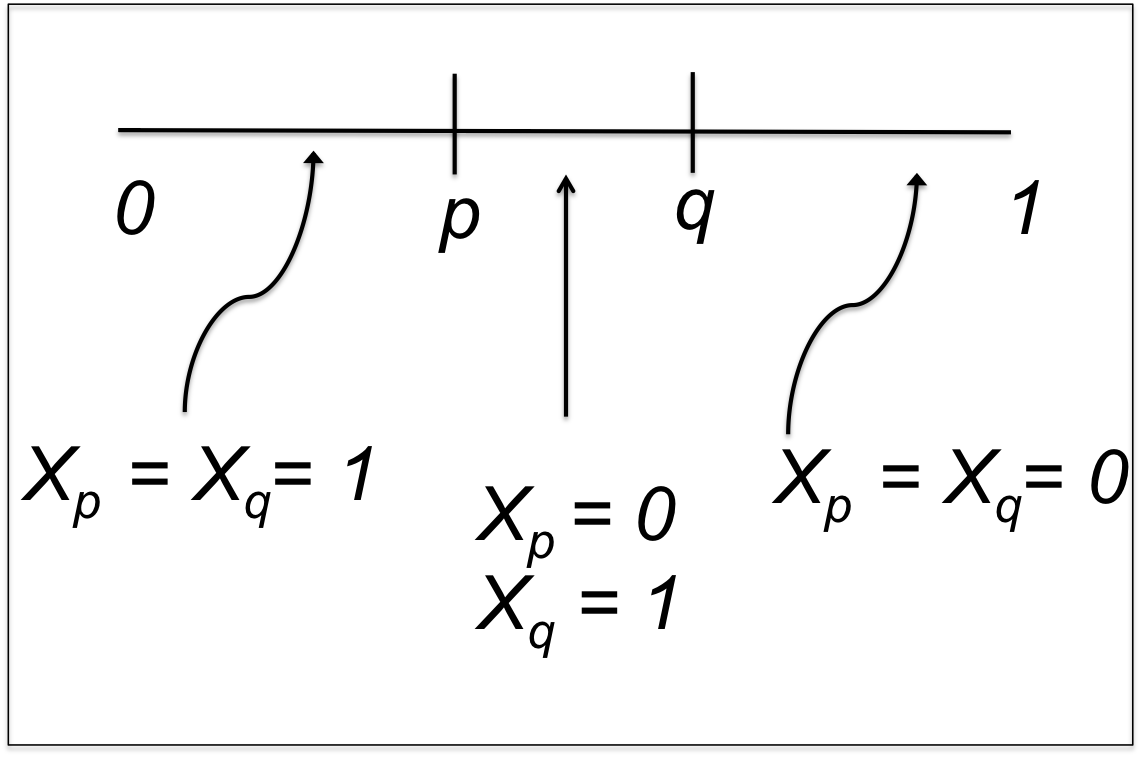}
\caption{Joint distribution of two Bernoulli random variables}
\label{bernoullidist}
\end{center}
\end{figure}

\end{example}
A simple, though somewhat fundamental, coupling is the following:
\begin{example}[Quantile coupling]\label{quantcoupling}
Let $X$ be a random variable with distribution function $F$, that is,
\[P(X \le x) =F(x), \;\; x \in \mathbb{R}.\]
Let $U$ be a uniform random variable on $[0,1]$. Then, for random variable $\hat{X}=F^{-1}(U)$,
 \[P(\hat{X}\le x)=P(F^{-1}(U)\le x)=P(U\le F(x))=F(x),\;\; x \in \mathbb{R},\]
that is, $\hat{X}$ is a copy of $X, \; \hat{X} =_dX$. Thus, letting $F$ run over the class of all distribution functions (using the same $U$), yields a coupling of all differently distributed random variables, the \textit{quantile coupling}. One can show that quantile coupling\index{coupling!quantile} consists of positively correlated random variables.
\end{example}
\subsection{Strassen's Theorem, coupling event inequality, and maximal coupling}
An important application of quantile coupling is in reducing a stochastic order\index{order!stochastic} between random variables to a pointwise (a.s.) order\index{order!pointwise}: Let $X$ and $X^\prime$ be two random variables with distribution functions $F$ and $G$, respectively. If there is a coupling $(\hat{X},\hat{X}^\prime)$ of $X$ and $X^\prime$ such that $\hat{X}$ is \textit{pointwise dominated} by $\hat{X}^\prime$, that is \[\hat{X} \le \hat{X}^\prime \;\; (\mbox{almost surely}),\]
then $\{\hat{X}\le x\} \supseteq \{\hat{X}^\prime\le x\}$,  which implies 
\[ P(\hat{X}\le x) \ge P(\hat{X}^\prime \le x),\]  and thus
 \[F(x) \ge G(x), \;\; x\in \mathbb{R}.\]  
Then $X$ is said to be \textit{stochastically} \textit{dominated} (or \textit{dominated in distribution}) by $X^\prime$: 
\[X \le_d X^\prime. \]

\noindent But one can show \citep[][p. 4]{thorisson2000} that the other direction also holds: \textit{stochastic} domination implies \textit{pointwise} domination.
Thus, we have a (simple) version of \textit{Strassen's theorem}\index{Strassen's theorem}\citep{strassen}:
\begin{theorem}[Strassen, 1965]
Let $X$ and $X^\prime$ be random variables. Then
\[X \le_d X^\prime\] if, and only if, there is a coupling $(\hat{X},\hat{X}^\prime)$ of $X$ and $X^\prime$ such that a.s. \[\hat{X} \le \hat{X}^\prime.\]
\end{theorem}

The following question is the starting point of many convergence and approximation results obtainable from coupling arguments. Let $X$ and $X^\prime$ be two random variables with non-identical distributions. How can one construct a coupling of $X$ and $X'$, $(\hat{X}, \hat{X}')$, such that $P(\hat{X}=\hat{X}')$ is maximal across all possible couplings? Here we follow the slightly more general formulation in \cite{thorisson2000} but limit presentation to the case of discrete random variables (the continuous case being completely analogous).

\begin{definition}\label{couplingevent}
Suppose $(\hat{X}_i : i \in \mathrm{I})$ is a coupling of $X_i, i \in \mathrm{I}$, and let $C$ be an event such that if it occurs, then all the $\hat{X}_i$ coincide, that is,
\[ C \subseteq \{\hat{X}_i=\hat{X}_j, \text{ for all } i,j \in \mathrm{I}\}.\] Such an event is called a \textit{coupling event}\index{coupling!event}.
\end{definition}
Assume all the $X_i$ take values in a finite or countable set $E$ with $P(X_i=x)=p_i(x)$, for $x \in E$. For all $i,j \in \mathrm{I}$ and $x \in E$, we clearly have
\[P(\hat{X}_i=x, C)=P(\hat{X}_j=x,C)\le p_j(x),  \]
and thus for all $i \in \mathrm{I}$ and $x \in E$,
\[ P(\hat{X}_i=x,C)\le \inf_{j \in \mathrm{I}}p_j(x). \]
Summing over $x \in E$ yields the basic \textit{coupling event inequality}\index{coupling!event!inequality}:
\begin{equation}\label{couplinginequal}
P(C) \le \sum_{x \in E}\, \inf_{j \in \mathrm{I}} \,p_i(x).
\end{equation}
As an example, consider again the case of two discrete random variables $X$ and $X'$ with coupling $(\hat{X}, \hat{X}')$, and set $C=\{\hat{X}=\hat{X}'\}$. Then 
\begin{equation} \label{twodiscrete}
P(\hat{X}=\hat{X}') \le \sum_{x} \min\{P(X=x),P(X'=x)\}. 
\end{equation}
Interestingly, it turns out that, at least in principle, one can always construct a coupling such that the above coupling event inequality (\ref{couplinginequal}) holds with identity. Such a coupling is called \textit{maximal}\index{coupling!maximal} and $C$ a \textit{maximal coupling event}.
\begin{proposition}(Maximal coupling)
Suppose $X_i, i \in \mathrm{I}$, are discrete random variables taking values in a finite or countable set $E$. Then there exists a maximal coupling, that is, a coupling with coupling event $C$ such that
\[ P(C)=\sum_{x \in E}\, \inf_{i \in \mathrm{I}} \,p_i(x). \]
\end{proposition}
\begin{proof}
Put
\[ c :=  \sum_{x \in E} \,\inf_{i \in \mathrm{I}} \,p_i(x) \;\; \text{ (the maximal coupling probability) }.\] If $c=0$, take the $\hat{X}_i$ independent and $C=\emptyset$. If $c=1$, take the $\hat{X}_i$ identical and $C=\Omega$, the sample space. For $0<c<1$, these couplings are mixed as follows:
Let $J$, $V$, and $W_i, ˜i \in \mathrm{I}$, be independent random variables such that $J$ is  Bernoulli distributed with $P(J=1)=c$ and, for $x \in E$, 
\begin{align*}
P(V=x)=&\frac{\inf_{i \in \mathrm{I}}\,p_i(x)}{c}\\
P(W_i=x)=&\frac{p_i(x)-c \,P(V=x)}{1-c}.
\end{align*}
Define, for each $i \in \mathrm{I}$,
\begin{equation}\label{splitting}
\hat{X}_i =\begin{cases}
V, \;\;\text{ if $J=1$,}\\
W_i, \;\;\text{ if $J=0$.}
\end{cases}
\end{equation}
Then
\begin{align*}
P(\hat{X}_i=x) &= P(V=x)P(J=1)+P(W_i=x)P(J=0)\\
			&=P(X_i=x).
\end{align*}
Thus, the $\hat{X}_i$ are a coupling of the $X_i$, $C=\{J=1\}$ is a coupling event, and it has the desired value $c$. 
\end{proof} 
The representation (\ref{splitting}) of the $X_i$ is known as \textit{splitting representation}\index{splitting representation}. 

We conclude the treatment of coupling with a variation on the theme of maximal coupling: Given two random variables, what is a measure of closeness between them when an appropriate coupling is applied to make them as close to being identical as possible?

The \textit{total variation distance}\index{total variation distance} between two probability distributions $\mu$ and $\nu$ on $\Omega$ is defined as
\[ \|\mu-\nu\|_{TV}:=\max_{A \subseteq \Omega} |\mu(A)-\nu(A)|, \]
for all Borel sets $A$. Thus, the distance between $\mu$ and $\nu$ is the maximum difference between the probabilities assigned to a single event by the two distributions. Using the coupling inequality, it can be shown that 
\begin{equation}
\|\mu-\nu\|_{TV}=\inf\{P(X\ne Y) \,|\; (X,Y) \text{ is a coupling of $\mu$ and $\nu$}\}.
\end{equation}
A splitting representation analogous to the one in the previous proof assures that a coupling can be constructed so that the infimum is obtained \citep[see][pp. 50--52]{levin2008}. 
\section{Copulas}
\cite{frechet51} studied the following problem, formulated here for the bivariate case: 
Given the distribution functions $F_X$ and $F_Y$ of two random variables $X$ and $Y$ defined on the same probability space, what can be said about the class $\mathcal{G}(F_X,F_Y)$ of the bivariate distribution functions whose marginals are $F_X$ and $F_Y$?

Obviously, the class $\mathcal{G}(F_X,F_Y)$ is non-empty since it contains the case of $X$ and $Y$  being independent. Let $F(x,y)$ be a joint distribution function for $(X,Y)$.  To each pair of real numbers $(x,y)$, we can associate three numbers: $F_X(x)$, $F_Y(y)$, and $F(x,y)$. Note that each of these numbers lies in the interval $[0,1]$. In other words, each pair $(x,y)$ of real numbers is mapped to a point $(F_X(x),F_Y(y))$ in the unit square $[0,1]\times[0,1]$, and this ordered pair in turn corresponds to a number $F(x,y)$ in $[0,1]$. 
\begin{align*}
(x,y)\mapsto (F_X(x),F_Y(y)) &\mapsto F(x,y) =C(F_X(x),F_Y(y)),\\
\mathbb{R} \times \mathbb{R} \longrightarrow [0,1] \times [0,1] &\overset{C}{\longrightarrow}  [0,1]
\end{align*}
Then the mapping $C$ is an example of a copula (it ``couples'' the bivariate distribution with its marginals).  
\subsection{Definition, Examples, and Sklar's Theorem}

A straightforward definition for any finite dimension $n$ is the following:
\begin{definition}
A function $C: [0,1]^n \longrightarrow [0,1] $ is called $n$-dimensional \linebreak\textit{copula}  if there is a probability space $(\Omega, \mathcal{F},\mathbb{P})$ supporting a vector of standard uniform random variables $(U_1,\ldots,U_n)$ such that
\[ C(u_1,\ldots,u_n)=P(U_1\le u_1,\ldots, U_n\le u_n),\;\; u_1,\ldots,u_n \in [0,1]. \]
\end{definition}
Clearly, any copula is a distribution function. There is an alternative, analytical definition of copula based on the fact that distribution functions can be characterized as functions satisfying certain conditions, without reference to a probability space. 
\begin{definition}
An $n$-dimensional copula $C$\index{copula!$n$-dimensional} is a function on the unit $n$-cube $[0,1]^n$ that  satisfies the following properties:
\begin{enumerate}
\item the range of $C$ is the unit interval $[0,1]$;
\item $C(\mathbf{u})$ is zero for all $\mathbf{u}$ in $[0,1]^n$ for which at least one coordinate is zero (groundedness)\index{groundedness};
\item $C(\mathbf{u})=u_k$ if all coordinates of  $\mathbf{u}$ are $1$ except the $k$-the one;
\item $C$ is \textit{$n$-increasing}, that is, for every $\mathbf{a}\le \mathbf{b}$ in $[0,1]^n$ ($\le$ defined componentwise) the volume assigned by $C$ to the $n$-box $[\mathbf{a},\mathbf{b}]=[a_1,b_1]\times\cdots \times [a_n,b_n]$ is nonnegative.
\end{enumerate}
\end{definition}
One can show that groundedness and the $n$-increasing property are sufficient to define a proper distribution function. Moreover, copulas are \textit{uniformly continuous} and all their \textit{partial derivatives} exist almost everywhere, which is a useful property especially for computer simulations \citep[for proofs see, e.g.,][]{Durante2016}. 

A very simple copula is the following:
\begin{example}[The independence copula]
For independent standard uniform random variables $U_1,\ldots,U_n$ and $\mathbf{U}=(U_1,\ldots,U_n)$
\[ P(\mathbf{U}\le \mathbf{u})= C(u_1,\ldots,u_n)=  \prod_{i=1}^{n}u_i \]
is a copula, called the \textit{independence copula}.
\end{example}
There are two further copulas of special importance:
\begin{example}[The comonotonicity copula]
\label{example:comono}
For $U$ uniformly distributed on $[0,1]$, consider the random vector $\mathbf{U}=(U,\ldots,U)$. Then, for any $\mathbf{u}\in [0,1]^n$,
\[  P(\mathbf{U}\le \mathbf{u})= P(U \le \min\{u_1,\ldots, u_n\})=\min\{u_1,\ldots, u_n\} \] 
is a copula, called the \textit{comonotonicity copula}.
\end{example}
\begin{example}[The countermonotonicity copula]
\label{example:counter}
For $U$ uniformly distributed on $[0,1]$, consider the random vector $\mathbf{U}=(U,1-U)$. Then, for any $\mathbf{u}\in [0,1]^2$,
\[  P(\mathbf{U}\le \mathbf{u})= P(U\le u_1, 1-U \le u_2)= \max\{0,u_1+u_2-1\} \]
is a copula, called the \textit{countermonotonicity copula}.
\end{example}
The comonotonicity copula is often denoted as \[M_n(u_1,\ldots, u_n)=\min\{u_1,\ldots, u_n\}\]
and is also called \textit{upper Fr\'{e}chet-Hoeffding bound copula}. Similarly, the function \[ W_n(u_1,\ldots, u_n)= \max\{u_1+\ldots+u_n - (n-1),0\} \] is called \textit{lower Fr\'{e}chet-Hoeffding bound copula} for $n=2$, but it is not a copula for $n>2$. The reason for the latter statement is that $W_n$ for $n\ge 3$ is in general not a proper distribution function (see below Section~\ref{section:depend}). Importantly, any copula obeys the \textit{Fr\'{e}chet-Hoeffding bounds}:
\begin{theorem}[\cite{frechet51}]\label{frechet51}
If $C$ is any $n$-dimensional copula, then for every $\mathbf{u}\in [0,1]^n$,
\[W_n(\mathbf{u}) \le C(\mathbf{u}) \le M_n(\mathbf{u}).  \]
\end{theorem}
Proof: see, e.g.,\citet{Durante2016}, p.27.

Although the Fr\'{e}chet-Hoeffding lower $W_n$ is never a copula for $n\ge 3$, it is the best possible lower bound in the following sense:
\begin{theorem}
For any $n\ge 3$ and any $\mathbf{u}\in [0,1]^n$, there is an $n$-dimensional copula (which depends on $\mathbf{u}$) such that 
\[ C(\mathbf{u})= W_n(\mathbf{u}) .\]
\end{theorem}
For the proof, see \citet{nelsen99}, p.48. The following famous theorem laid the foundation of many subsequent studies \citep[for a proof, see][Theorem~2.10.9]{nelsen99}.
\begin{theorem}[Sklar's Theorem, 1959]
Let $F(x_1,\ldots, x_n)$ be an $n$-variate distribution function  with margins $F_1(x_1), \ldots, F_n(x_n)$; then there exists an $n$-copula $C:[0,1]^n \longrightarrow [0,1]$ that satisfies
\[F(x_1, \ldots, x_n) = C(F_1(x_1),\ldots, F_n(x_n)), \;\;\; (x_1, \ldots x_n) \in \mathbb{R}^n.\]
If all univariate margins $F_1, \ldots, F_n$ are continuous,  then the copula is unique. Otherwise, $C$ is uniquely determined on $\mathrm{Ran}F_1 \times \mathrm{ Ran}F_2 \times \ldots  \mathrm{Ran}F_n$.

If $F_1^{-1}, \ldots, F_n^{-1}$ are the quantile functions of the margins, then for any $(u_1, \ldots,u_n) \in [0,1]^n$
\[C(u_1, \ldots,u_n)=F(F_1^{-1}(u_1),\ldots,F_n^{-1}(u_n)).\] 
\end{theorem}
Copulas with discrete margins have also been defined, but their treatment is less straightforward \cite[for a review, see][]{pfeifer2004}. Sklar's theorem shows that copulas remain invariant under strictly increasing transformations of the underlying random variables. It is possible to construct a wide range of multivariate distributions by separately choosing the marginal distributions and a suitable copula. 
\begin{example}(Bivariate exponential)
For $\delta >0$, the distribution
\[ F(x,y)=\exp\{-[e^{-x}+e^{-y}-(e^{\delta x}+e^{\delta y})^{-1/\delta}]\},\;\;\; -\infty<x,y,<+\infty, \]
with margins $F_1(x)=\exp\{-e^{-x}\}$ and $F_2(y)=\exp\{-e^{-y}\}$ corresponds to the copula 
\[C(u,v)=u v \exp\{[(-\log u)^{-\delta}+(-\log v)^{-\delta}]^{-1/\delta}\},  \]
an example of the class of \textit{bivariate extreme value copulas} characterized by
$C(u^t,v^t)=C^t(u,v)$, for all $t>0$.
\end{example}
\subsection{Copula density and pair copula constructions (vines) }
If the probability measure associated with a copula $C$ is absolutely continuous (with respect to the Lebesgue measure on $[0,1]^n$), then there exists a \textit{copula density} $c: [0,1]^n \longrightarrow [0,\infty]$ almost everywhere unique such that 
\[ C(u_1,\ldots,u_n)= \int\limits_{0}^{u_1} \cdots \int\limits_{0}^{u_n}\, c(v_1,\ldots,v_n) \,dv_n \ldots dv_1, \;\; u_1,\ldots,u_n \in [0,1].\]
Such an absolutely continuous copula is $n$ times differentiable and
\[ c(u_1,\ldots,u_n)= \frac{\partial}{\partial u_1}\cdots  \frac{\partial}{\partial u_n}\,C(u_1,\ldots,u_n) ,\;\; u_1,\ldots,u_n \in [0,1].\]
For example, the independence copula is absolutely continuous with density equal to $1$:
\[\Pi (u_1,\ldots,u_n)=\prod\limits_{k=1}^{n} u_k=\int\limits_{0}^{u_1} \cdots \int\limits_{0}^{u_n} \,1 \,dv_n \ldots dv_1.  \]
When the density of a  distribution $F_{12}(x_1,x_2)$ exists, differentiating  yields 
\[ f_{12}(x_1,x_2)=f_1(x_1)f_2(x_2) \,c_{12}(F_1(x_1),F_2(x_2)) .\]
This equation shows how independence is ``distorted'' by copula density $c$ whenever $c$ is different from $1$. Moreover, this yields an expression for the conditional density of $X_1$ given $X_2=x_2$:
\begin{equation} \label{cond}
 f_{1|2} (x_1| x_2)= c_{12}(F_1(x_1),F_2(x_2)) f_1(x_1)
\end{equation}
This the starting point of a recent, important approach to constructing high-dimensional dependency structures from pairwise dependencies (``vine copulas''). Note  that  a multivariate density of dimension $n$ can be decomposed as follows, here taking the case for $n=3$:
\[ f (x_1,x_2,x_3)= f_{3|12}(x_3|x_1,x_2) f_{2|1}(x_2|x_1) f_1(x_1).\]
Applying the decomposition in Equation~\ref{cond} to each of these terms yields,
\begin{align*}
f_{2|1}(x_2| x_1) &= c_{12}(F_1(x_1), F_2(x_2)) f_2(x_2) \\
f_{3|12}(x_3|x_1,x_2) &= c_{13|2}(F_{1|2}(x_1|x_2),F_{3|2}(x_3|x_2)) f_{3|2}(x_3 |x_2) \\
f_{3|2}(x_3| x_2) &= c_{23}(F_2(x_2), F_3(x_3)) f_3(x_3),
\end{align*}
resulting in the ``regular vine tree'' representation
\begin{align}\label{vinetree}
f(x_1,x_2,x_3) &= f_3(x_3) f_2(x_2)f_1(x_1)  \;\; \text{(marginals)} \\
						& \times c_{12}(F_1(x_1), F_2(x_2)) \cdot c_{23}(F_2(x_2), F_3(x_3)) \; \text{(unconditional pairs)} \notag \\
						& \times c_{13|2}(F_{1|2}(x_1|x_2),F_{3|2}(x_3 |x_2)) \; \text{(conditional pair).} \notag
\end{align}
In order to visualize this structure, in particular for larger $n$, one defines a sequence of trees (acyclic undirected graphs), a simple version of it is depicted in Figure~\ref{vine}.

\begin{figure}[htbp]
\begin{center}
\includegraphics[scale=0.68]{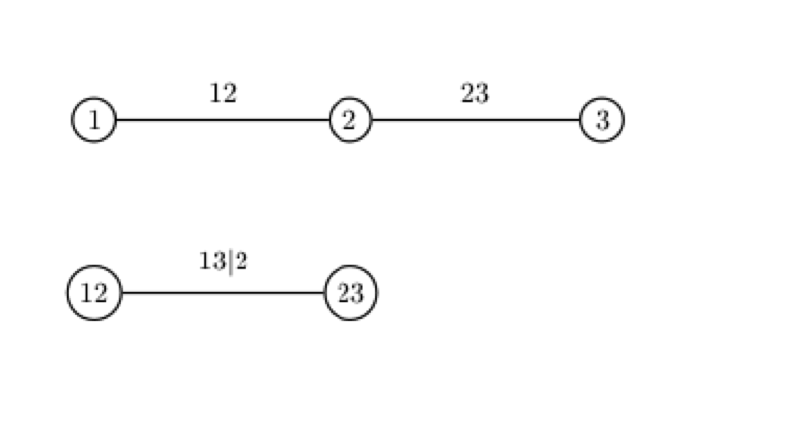}
\caption{Graphical illustration of the decomposition in Equation~\ref{vinetree}. A line in the graph corresponds to the indices of a copula linking two distributions, unconditional in the upper graph, conditional in the lower graph.}
\label{vine}
\end{center}
\end{figure}
\subsection{Survival copula, co-copula, dual and diagonal section of a copula}
Whenever it is more convenient to describe the multivariate distribution of a random vector $(X_1,\ldots,X_n)$ by means of its survival distribution, i.e.,
\[ \hat{F}(x_1,\ldots,x_n)= P(X_1>x_1,\ldots, X_n>x_n),\]
its \textit{survival copula} can be introduced such that the analog of Sklar's theorem holds, with
\[  \hat{F}(x_1,\ldots,x_n)= \hat{C}(\hat{F}_1(x_1),\ldots, \hat{F_n}(x_n)),  \]
where the $\hat{F}_i$, $i=1,\ldots,x_n$, are the marginal survival distributions. In the continuous case there is a one-to-one correspondence between the copula and its survival copula. For $n=2$, this is
\[ C(u_1,u_2)= \hat{C}(1-u_1,1-u_2) +u_1 + u_2-1.\]
For the general case, we refer to \cite[pp. 20-21]{maischerer2012}.

Two other functions closely related to copulas and survival copulas are useful in the response time modeling context. The \textit{dual of a copula $C$} is the function $\tilde{C}$ defined by $\tilde{C}(u,v) = u + v - C(u,v)$ and the \textit{co-copula} is the function $C^*$ defined by $C^*(u,v) =1- C(1- u,1- v)$. Neither of these is a copula, but when $C$ is the (continuous) copula of a pair of random variables $X$ and $Y$, the dual of the copula and the co-copula each express a probability of an event involving $X$ and $Y$:
\[ \tilde{C}(F_X(x),F_Y(y))= P(X\le x \, \text{ or }\, Y\le y) \]
and 
\[ C^*(1-F_X(x),1-F_Y(y)) = P(X > x \, \text{ or }\, Y > y). \]
Finally, for standard uniform random variables $U_1,\ldots, U_n$ and the corresponding copula $C(u_1,\ldots,u_n)$, its \textit{diagnonal section} is defined as $\delta_C(u)=C(u, \ldots, u)$. \citet[][p. 69]{Durante2016} state necessary and sufficient conditions for a function $\delta: [0,1] \rightarrow [0,1]$ to be the diagonal section of some copula.
\subsection{Copulas with singular components}
\label{singular}
If the probability measure associated with a copula $C$ has a singular component, then the copula also has a singular component which can often be detected by finding points $(u_1,\ldots,u_n) \in [0,1]^n$, where some (existing) partial derivative of the copula has a point of discontinuity. A standard example is the comonotonicity copula
\[ M_n(u_1,\ldots,u_n)=\min(u_1,\ldots, u_n),\]
where the partial derivatives have a point of discontinuity;
\[  \frac{\partial}{\partial u_k} M_n(u_1,\ldots,u_n)=\begin{cases} 1 , &\text{if $u_k < \min(u_1,\ldots u_{k-1},u_{k+1},\ldots, u_n)$,}\\
0, &\text{if  $u_k > \min(u_1,\ldots u_{k-1}, u_{k+1},\ldots, u_n)$.} \end{cases} \]
The probability measure associated with $M_n(u_1,\ldots,u_n)$ assigns all mass to the diagonal of the unit $n$-cube $[0,1]^n$ (``perfect positive dependence'').
\subsection{Copulas and extremal dependence}
\label{section:depend}
Here we take a closer look at how copulas relate to stochastic dependency. For $n=2$, Theorem~\ref{frechet51} reduces to 
\begin{example}[Fr\'{e}chet-Hoeffding copula]\label{Hoeffding-copula} %
Let $C(u,v)$ be a $2$-copula; then, for $u, v \in [0,1]$,
\[ W_2(u,v) \equiv \max\{u+v-1,0\} \le C(u,v) \le \min\{u,v\}\equiv M_2(u,v),\]
and $M$ and $W$ are also copulas, the \textit{upper} and \textit{lower}  \textit{Fr\'{e}chet-Hoeffding  copula}.
\end{example}
We have seen (Examples~\ref{example:comono} and \ref{example:counter}) that $M_2$ and $W_2$ are bivariate distribution functions of the random vectors $(U,U)$ and $(U,1-U)$, respectively, where $U$ is a standard uniform random variable. In this case, we say that $M_2$ (comonotonicity copula) describes \textit{perfect positive dependence} and $W_2$ (countermonotonicity copula) describes \textit{perfect negative dependence}. For $U$ and $V$ standard uniform random variables whose joint distribution is the copula $M_2$, then $P(U=V)=1$; and if the copula is $W_2$, then $P(U+V=1)=1$.

If $X$ and $Y$ are random variables with joint distribution function $H(x,y)$ and margins $F_X(x)$ and $F_Y(y)$, then it is easy to show \citep[e.g.] [p. 47]{joe2015} that, for all $x, y \in \mathbb{R},$ 
\begin{equation}\label{f-bound}
\max\{F_X(x)+F_Y(y)-1,0\}\le H(x,y) \le \min\{F_X(x),F_Y(y)\}
\end{equation}
$F^-(x,y)=\max\{F_X(x)+F_Y(y)-1,0\}$ is called the \textit{lower} and $F^+(x,y)=\min\{F_X(x),F_Y(y)\}$ the \textit{upper Fr\'{e}chet-Hoeffding bound}, respectively, or \textit{Fr\'{e}chet bound}, for short. What can be said about random variables $X$ and $Y$ when their joint distribution $H$ equals one of its Fr\'{e}chet-Hoeffding bounds?

If both margins $F_X$ and $F_Y$ are continuous then, by Sklar's theorem, the copulas corresponding to $H$ are unique and the Fr\'{e}chet bounds $F^+$ and $F^-$ represent perfect positive and negative dependence, respectively. In the discrete case, the bounds can also sometimes represent perfect dependence, but not with any generality \citep[see Example 2.9 in] []{joe2015}.
\subsection{Linear measures of dependence}
The most widely known and used dependence measure is Pearson's linear correlation,
\[ \rho(X,Y)= \frac{\mathrm{Cov}(X,Y)}{\sqrt{\mathrm{Var}(X)\mathrm{Var}(X)}}. \]
Given finite variances, it is a measure of linear dependence that takes values in the range $[-1,1]$. An obvious disadvantage of $\rho$ in the context of copulas, in addition to requiring finite variances, is that it depends on the marginal distributions; thus, it is not invariant under strictly increasing nonlinear transformations of the variables, whereas copulas are. \cite{Embrechts2002} mention two ``pitfalls'' in dealing with linear correlation: (1)~assuming that marginal distributions and correlation determine the joint distributions, and (2)~assuming that, given marginal distributions $F_X$ and $F_Y$, all linear correlations between $-1$ and $1$ can be attained through suitable specification of the joint distribution of $X$ and $Y$. Counterexamples to both assumptions abound in the copula literature \citep[e.g.,][]{joe2015,Embrechts2003,Embrechts2002,nelsen99}.

\begin{proposition}[\citealp{hoeffding40}]\label{Prop:hoeffding}
Let $X$ and $Y$ have finite (nonzero) variances with an unspecified dependence structure. Then
\begin{enumerate}
\item the set of possible linear correlations is a closed interval $[\rho_{\min},\rho_{\max}]$ and for the extremal correlations $\rho_{\min} < 0 < \rho_{\max}$ holds;
\item the extremal correlation $\rho=\rho_{\min}$ is
attained if and only if $X$ and $Y$ are countermonotonic; $\rho = \rho_{\max}$ is attained if and only if $X$ and $Y$ are comonotonic.
\item $\rho_{\min}=-1$ if and only if $X$ and $-Y$ are of the same type and $\rho_{\max}=1$ if and only if $X$ and $Y$ are of the same type ($X$ and $Y$ are \emph{of the same type} if we can find $a>0$ and $b\in \mathbb{R}$ so that $Y=_d aX+b$).
\end{enumerate}
\end{proposition}
The proof of this proposition presented in \citealp[][(p. 24)]{Embrechts2002} starts by recalling the Fr\'{e}chet-Hoeffding bounds (Equation~\ref{f-bound}),
\begin{equation*}
F^-=\max\{F_X(x)+F_Y(y)-1,0\}\le H(x,y) \le \min\{F_X(x),F_Y(y)\}=F^+.
\end{equation*}
Inserting the upper and the lower bound into  Hoeffding's identity (e.g. \citealp{shea83})
\begin{equation}\label{hoeff_id}
\mathrm{Cov}(X,Y)=\int\limits_{-\infty}^{+\infty} \int\limits_{-\infty}^{+\infty}\left[H(x,y)-F_X(x)F_Y(y)\right]\,dx\,dy ,
\end{equation}
immediately gives the set of possible correlations (see \citealt[][(ibid.)]{Embrechts2002} for the complete proof of the proposition). 

With $0\le \lambda \le 1$, the mixture $\lambda F^- +(1-\lambda) F^+$ has correlation \[\rho=\lambda \rho_{\min} +(1-\lambda) \rho_{\max}\] and can be used to construct joint distributions with marginals $F_X$ and $F_Y$ and with arbitrary correlations $\rho \in [\rho_{\min},\rho_{\max}]$.

Finally, the following example shows that small (linear) correlations cannot be interpreted as implying weak dependence between random variables.
\begin{example}[\citet{Embrechts2002}]
Let $X$ be distributed as $\mathrm{Lognormal}(0,1)$ and $Y$ as $\mathrm{Lognormal}(0,\sigma^2)$, $\sigma >0$. Note that $X$ and $Y$ are not of the same type although $\log X$ and $\log Y$ are. From Proposition~\ref{Prop:hoeffding}, one obtains
\[ \rho_{\min}=\frac{e^{-\sigma}-1}{\sqrt{(e-1)(e^{\sigma^2}-1)}} \] and 
\[ \rho_{\max}=\frac{e^{\sigma}-1}{\sqrt{(e-1)(e^{\sigma^2}-1)}}. \]
Letting $\sigma \rightarrow \infty$, both correlations converge to zero. 
\end{example}
Thus, it is possible to have a random vector $(X,Y)$ where the correlation is almost zero, even though $X$ and $Y$ are comonotonic or countermonotonic and therefore have the strongest kind of dependency possible.
\subsection{Copula-based measures of dependence}
There are several alternatives to the linear correlation coefficient when the latter is inappropriate or misleading. Two important ones are \textit{Kendall's tau} and \textit{Spearman's rho}. 

For a vector $(X,Y)$, Kendall's tau is defined as
\[ \tau(X,Y)=P[(X-\tilde{X})(Y-\tilde{Y})>0]-P[(X-\tilde{X})(Y-\tilde{Y})<0], \]
where $(\tilde{X},\tilde{Y})$ is an independent copy of $(X,Y)$. Hence Kendall's tau is simply the probability of concordance minus the probability of discordance. When $(X,Y)$ is continuous with copula $C$, then it can be shown that
\[ \tau(X,Y)=\tau_C = 4 \iint\limits_{[0,1]^2} C(u,v)dC(u,v) -1. \]
Note that the integral above can be interpreted as the expected value of the function $C(U,V)$ of the standard uniform  random variables $U$ and $V$ whose joint distribution is $C$, i.e., 
\[\tau_C = 4\,\mathrm{E}[C(U,V)]-1. \]
In the case of  three known bivariate distribution functions, Kendall's tau  gives a necessary condition for compatibility, i.e., for the existence of a trivariate distribution function with the given bivariate marginals. Specifically,
\begin{proposition}\cite[][p. 76]{joe1997}
Let $F \in \mathcal{G}(F_{12},F_{13},F_{23})$, the class of trivariate distributions with marginals $F_{12},F_{13},F_{23}$ and suppose $F_{jk}$, $j<k$, are continuous. Let $\tau_{jk}=\tau_{kj}$ be the value of Kendall's tau for $F_{jk}$, $j\ne k$. Then the inequality
\[ -1+|\tau_{ij}+\tau_{jk}| \le \tau_{ik} \le 1-|\tau_{ij}-\tau_{jk}|\]
holds for all permutations $(i,j,k)$ of $(1,2,3)$ and the bounds are sharp.
\end{proposition}
Thus, if the above inequality does not hold for some $(i,j,k)$, then the three bivariate margins are not compatible. Sharpness follows from the special trivariate normal case: Kendall's tau for the bivariate normal is $\tau=(2/\pi)\arcsin(\rho)$, so that the inequality becomes
\[ -\cos(\tfrac{1}{2}\pi (\tau_{12}+\tau_{23}))\le \sin(\tfrac{1}{2} \pi \tau_{13})\le \cos(\tfrac{1}{2} \pi(\tau_{12}-\tau_{23})), \]
with $(i,j,k)=(1,2,3)$.

Let us now turn to Spearman's rho. For three independent vectors $(X_1,Y_1), (X_2,Y_2)$ and $(X_3,Y_3)$ with a common joint distribution $H$ (whose margins are $F$ and $G$), Spearman's rho is is defined as 
\[\rho(X,Y)= 3 \{ P[(X_1-X_2)(Y_1-Y_3)>0]-P[(X_1-X_2)(Y_1-Y_3)<0] \}. \]
Note that $(X_1,Y_1)$ and $(X_2,Y_3)$ is pair of vectors with the same marginals, but $(X_1,Y_1)$ has distribution function $H$, while the components of $(X_2,Y_3)$ are independent. For $X$ and $Y$ continuous with copula $C$, one can show \citep{nelsen99} that
\[ \rho(X,Y)=\rho_C=12\,\iint\limits_{[0,1]^2} uv\; dC(u,v) -3. \]
and, moreover, that Spearman's rho is identical to the linear correlation coefficient between the random variables $U=F(X)$ and $V=G(Y)$:
\begin{align*}
\rho(X,Y)=\rho_C  = & 12 \;\mathrm{E}[UV]-3\\
                = & \frac{\mathrm{E}[UV]-1/4}{1/12}=\frac{\mathrm{E}[UV]-\mathrm{E}[U]\mathrm{E}[V]}{\sqrt{\mathrm{Var}[U]}\sqrt{\mathrm{Var}[V]}}.
\end{align*}
Finally, the relation between Kendall's tau and Spearman's rho has been investigate early on. It has been shown \cite[e.g.][]{kruskal58}  that always 
\[-1 \le 3 \tau -2 \rho \le 1,\]
where $\tau$ is Kendall's tau and $\rho$ is Spearman's rho.
\section{Example application: multisensory modeling}
While the applications discussed here refer to the context of multisensory modeling, similar examples can more generally be found in any context where processing of information takes place in two or more channels (e.g. the visual search paradigm mentioned in the introduction). Moreover, we do not strive to be exhaustive in presenting multisensory applications, nor do we try to treat them in depth. Rather, our goal is to encourage further work in an area that is quickly gaining importance in psychology, neuroscience, and related areas.

In the behavioral version of the \textit{multisensory paradigm}, one distinguishes two different conditions:
\begin{quote}
\textit{Unimodal} condition: a stimulus of a single modality (visual, auditory, tactile) is presented and the participant is (i)~asked to respond (e.g., by button press or eye movement) as quickly as possible upon detecting the stimulus (\textit{reaction time task}), or (ii)~to indicate whether or not a stimulus of that modality was detected (\textit{detection task}). 
\end{quote}

\begin{quote}
\textit{Bi- or trimodal} condition: stimuli from two or three modalities are presented (nearly) simultaneously and the participant is (i)~asked to respond as quickly as possible upon detecting a stimulus of any modality (\textit{redundant signals RT task}), or (ii)~ to indicate whether or not a stimulus of any modality was detected (\textit{redundant signals detection task})
\end{quote}
We refer to $\mathcal{V}, \mathcal{A}, \mathcal{T}$ as the unimodal context where visual, auditory, or tactile stimuli are presented, resp. Simlarly, $\mathcal{VA}$ denotes a bimodal (visual-auditory) context, etc. For each stimulus, or crossmodal stimulus combination, we observe samples from a random variable representing the reaction time measured in any given trial. Let $F_V(t), F_A(t), F_{VA}(t)$ denote the (theoretical) distribution functions of reaction time in a unimodal visual, auditory, or a visual-auditory context, respectively, when a specific stimulus (combination) is presented\footnote{For simplicity, we write $F_V(t)$, etc., instead of $F_\mathcal{V}(t)$.}. Analogously, we define the probabilities for indicator functions in a detection task: $p_V=P(detection|V)$, $p_A=P(detection|A)$, and  $p_{VA}=P(detection|VA)$.

Note that, from a modeling point of view, each context $\mathcal{V}, \mathcal{A}$, or $\mathcal{VA}$ refers to a different sample space and $\sigma$-algebra. Therefore,  no probabilistic coupling between the (reaction time) random variables in these different conditions necessarily exist. A common assumption, often not stated explicitly, is that there does exist a coupling between visual and auditory RT, for example, such that the margins of the coupling, i.e., of a bivariate distribution $\hat{H}_{VA}$, are equal to the distributions $F_V$ and $F_A$. 

Assuming such a coupling exists, a multisensory model should specify how $F_{VA}$ relates to the bivariate distribution $\hat{H}_{VA}$. In principle, this could be tested empirically. However, in the multisensory RT paradigm described above, the marginals of $\hat{H}_{VA}$ are not observable, only the distribution function of RTs in the bimodal context, $F_{VA}$, is. Therefore, testing for the existence of a coupling requires an additional assumption about how $\hat{H}_{VA}$ and $F_{VA}$ are related. The model studied most often is the so-called \textit{race model}.
\begin{example}[The race model] Let $V$ and $A$ be the random reaction times in unimodal conditions $\mathcal{V}$ and $\mathcal{A}$, with distribution functions $F_V(t)$ and $F_A(t)$, resp. Assume a coupling exists, i.e., a bivariate distribution function $\hat{H}_{VA}$ for $(\hat{V},\hat{A})$ such that
 \[V=_d \hat{V} \mbox{ and } A=_d \hat{A};\] 
assume bimodal RT is determined by the ``winner'' of the race between the modalities:
\begin{center}
$F_{VA}(t)=P(\hat{V}\le t \mbox{ or } \hat{A} \le t)$
\end{center}
Then  
\begin{equation}\label{race}
F_{VA}(t)= F_V(t)+F_A(t) - \hat{H}_{VA}(t,t).
\end{equation} 
\end{example}
\noindent The function $\hat{H}_{VA}(s,t)=C(F_V(s),F_A(t))$ is clearly a copula and, from Sklar's theorem, it is unique assuming continuous unimodal distribution functions. Moreover, we have the upper and lower Fr\'{e}chet copulas (Example~\ref{Hoeffding-copula}) such that
\begin{equation}
\max\{F_V(s)+F_A(t)-1,0\}\le \hat{H}_{VA}(s,t) \le \min\{F_V(s),F_A(t)\}.
\end{equation}
Taking the diagonal sections of these copulas (i.e., setting $s=t$ throughout), inserting $F_{VA}(t)$ from Equation~\ref{race}, and rearranging yields the ``race model inequality'' \citep{miller82}:
\[\max\{F_V(t),F_A(t)\}\le F_{VA}(t) \le \min\{F_V(t)+F_A(t),1\}, \;t\ge 0.\]
The upper bound corresponds to maximal negative dependence between $\hat{V}$ and $\hat{A}$, the lower bound to maximal positive dependence. Empirical violation of the upper bound (occurring only for small enough $t$) is interpreted as evidence against the race mechanism (``bimodal RT faster than predictable from unimodal conditions''), but it may also be evidence against the coupling assumption\citep{colonius06}.

\begin{example}[Time window of integration model \citealt{Colonius2004}]
The time window of integration (\emph{TWIN}) model distinguishes a first stage where unimodal neural activations race against each other, and a subsequent stage of converging processes that comprise neural integration of the input and preparation of a response. Multisensory integration occurs only if all peripheral processes of the first stage terminate within a given temporal interval, the ``time window of integration''.  Total reaction time in the crossmodal (visual-auditory) condition is the sum of first and second stage processing times: 
\begin{equation}\label{sum}
RT_{VA} = W_1 + W_2,
\end{equation}
\noindent where $W_1$ and $W_2$ are two random variables on the same probability space. Letting $I$ denote the random event that integration occurs and $I^{c}$  its complement, $W_{1}$ and $W_{2}$ are assumed to be \emph{conditionally independent}, conditioning on either $I$ or $I^{c}$. This implies that any dependency between the processing stages is solely generated by the event of integration or its complement. The distribution of the pair $(W_1,W_2)$ then is
\begin{equation}\label{bivar}
H(w_{1},w_{2})=\pi H_{I}(w_{1},w_{2})+(1-\pi) H_{I^{C}}(w_{1},w_{2}),
\end{equation}
where $H_{I}$ and $H_{I^{C}}$ denote the conditional distributions of $W_{1}$ and $W_{2}$ with respect to $I$ and $I^{c}$, respectively, and $\pi=P(I)$. 
\end{example}
By conditional independence, $H_{I}$ and $H_{I^{C}}$ can be written as products of their marginal distributions,
\[H_{I}(w_{1},w_{2})=F_{I}(w_{1}) G_{I}(w_{2}) \mbox{ and } H_{I^{C}}(w_{1},w_{2})=F_{I^{C}}(w_{1}) G_{I^{C}}(w_{2}), \]
where $F$ and $G$ refer to the first and second stage (conditional) distributions, respectively. Inserting into \emph{Equation~\ref{bivar}} yields
\begin{equation}
\label{mix}
H(w_{1},w_{2})=\pi F_{I}(w_{1}) G_{I}(w_{2}) +(1-\pi) F_{I^{C}}(w_{1}) G_{I^{C}}(w_{2}).
\end{equation}
\noindent The covariance between $W_1$ and $W_2$, computed using Hoeffding's identity (Equation~\ref{hoeff_id}), equals 
\begin{equation}
\mathrm{Cov}(W_1,W_2)=\pi(1-\pi)\{\mathrm{E}(W_1|I^C)-\mathrm{E}(W_1|I)\}\{\mathrm{E}(W_2|I^C)-\mathrm{E}(W_2|I)\},
\end{equation}
showing that the dependence between the stage processing times can be positive, negative,  or zero\footnote{Investigation of Kendall's tau and Spearman's rho for this model will be pursued elsewhere.}.
\section*{Bibliographic and historical notes}
The origins of the theory of probabilistic coupling have been traced back to the work of Wolfgang D\"oblin\footnote{Son of Alfred D\"oblin (1878-1957), an important German writer well known for his novel \emph{Berlin Alexanderplatz}.} in the late 1930s \citep[see][]{doeblin1938,Lindvall1991}, but interest in the theory waxed and waned for a long time, and it has only recently become part of some standard texts in probability theory \citep[e.g.][]{gut13}. For an advanced treatment of coupling theory see \cite{lindvall2002,thorisson2000,jaworski2010}, and \citet{levin2008}.

According to \cite{Durante2010,Durante2016}, the history of copulas may  be traced back to \cite{frechet51}. However, the term \textit{copula} and the theorem bearing his name was introduced by Abe Sklar \citep{Sklar1959}. Comprehensive treatments of copula theory are  \cite{joe1997,joe2015,Durante2016,nelsen99}, \cite{maischerer2012} focus on the simulation aspects, and \cite{denuit2005} and \cite{ruesch2013} emphasize the actuarial and financial risks background of the theory. A compact and application oriented introduction is \cite{Trivedi2005}. 

To our knowledge, the first application of the concepts of coupling and copula (without using those terms) to reaction time modeling is \cite{colonius1990}. Early investigations of the race model inequality include \cite{miller82,Ulrich1986,Colonius1986,Diederich1987, Diederich1992}. An application of coupling/copula concepts to multisensory detection is \cite{colonius2015}.

\section*{Acknowledgments}
This work was supported by DFG (German Science Foundation) SFB/TRR-31 (Project B4) and the  DFG Cluster of Excellence EXC 1077/1 “Hearing4all.” 

\section*{References}

\bibliography{colonius_cop-refs}

\end{document}